\documentclass{journal}

\usepackage[utf8]{inputenc}
\usepackage{amsmath, amssymb, amsthm}
\usepackage[style = numeric-comp, sorting = none]{biblatex}

\theoremstyle{definition}
\newtheorem{dfn}{Definition}[section]
\newtheorem{exl}[dfn]{Example}

\theoremstyle{plain}
\newtheorem{lem}[dfn]{Lemma}

\newtheorem{cor}[dfn]{Corollary}

\theoremstyle{remark}
\newtheorem*{rmk*}{Remark}
\newtheorem*{clm*}{Claim}
\newtheorem*{not*}{Notation}

\numberwithin{equation}{section}

\newcommand{\impossible}{\ensuremath{-\infty}}
\newcommand{\set}[1]{\ensuremath{\left\{ #1 \right\}}}
\newcommand{\cl}{\ensuremath{\mathrm{cl}}}

\newcommand{\R}{\ensuremath{\mathbb{R}}}
\newcommand{\Rp}{\ensuremath{\mathbb{R}_{>0}}}

\newcommand{\Rimp}{\ensuremath{\mathbb{R}\cup\set{\impossible}}}

\title{A mathematical definition of property rights in a Debreu economy}
\author{Abhimanyu Pallavi Sudhir}

\begin{document}

\maketitle
\begin{abstract}
    We present mathematical definitions for rights structures, government and non-attenuation in a generalized $N$-person game (Debreu abstract economy), thus providing a formal basis for property rights theory.
\end{abstract}

\section{Introduction}
\label{sec:introduction}

The development of property rights theory goes back to Coase \cite{coase}, who is credited with the development of the ``property rights approach'' to institutional economics. This approach has been variously reviewed by \cite{randall, furubotn, jongwook, mahoney, antras}, and has been employed as a foundational basis for methodologically individualistic social science theories such as in political science \cite{downs, buchanan, olson} and organizational economics \cite{grossman}.

According to \cite{randall, furubotn}, property rights theory is the logical development of microeconomic theory to institutional economics. Despite this microfoundational role, however, property rights theory remains formulated only in the setting of an exchange economy, which prevents its generalization to a broader class of social sciences. In this paper we shall generalize the theory and its important results to an abstract economy.

We employ the following specification of Debreu's model \cite{debreu} in this paper: there is a collection of agents $A$, and each $\alpha\in A$ has a choice set $X_\alpha$ and utility function $U_\alpha:X\to\Rimp$ (where $X=\prod_\beta X_\beta$ is called a choice space). If a choice set can be written as a product of factors, then the factors (i.e. axes) are called \emph{goods}.

(Instead of requiring an action correspondence, we represent impossible choices as having $\impossible$ utility and call $\set{x\mid U_\alpha(x)\in\R}$ the \emph{support} of $U_\alpha$. Standard assumptions of the continuity of the utility function and on basic geometric properties of the choice set should then be replaced by suitable statements involving the support of $U_\alpha$.)

\section{Definitions}
\label{sec:definitions}

Intuitively, one might imagine aggression as an action that reduces another agent's utility -- however, this is only meaningful in a relative sense. The concepts of consent and aggression are only defined relative to a specified ``rights structure''.

\begin{not*}
    Tuples and sequences are written with parentheses, like $(x_\alpha), (x_n)$; the power set of $S$ is denoted as $2^S$; co-ordinates are labeled by subscripting, e.g. if $x\in X_\alpha\times X_\beta$, then we might write its co-ordinates as $x_\alpha\in X_\alpha, x_\beta\in X_\beta$; similarly the tuple of co-ordinates not corresponding to $\alpha$ may be denoted $x_{-\alpha}$. Unless otherwise specified, uppercase Latin letters denote choice sets and subsets thereof, lowercase Latin letters denote choices $x,y,\dots$ or indices $n, m, \dots$, lowercase Greek letters $\alpha,\beta,\dots$ denote agents.
\end{not*}

\begin{dfn}[Consentification]
    \label{dfn:consent}
    Let $X^0$ be a choice space, and define a sequence $(X^n)$ defined by the following recurrence rule:
    \begin{equation*}
        X^{n+1}_\alpha = X^0_\alpha\times \prod_\beta 2^{X^n_\beta}
    \end{equation*}
    (The $\beta$ co-ordinate of a choice $x_\alpha^n$ is denoted as $x_{\alpha,\beta}^{n}$.) And define for $m\le n$ the projections $\pi^{mn}_\alpha:X^n_\alpha\to X^m_\alpha$ through composition on the following recurrence:
    \begin{align*}
        \pi^{01}_\alpha(x_\alpha^0,R_{\alpha,-\alpha}^{0})&=x_\alpha^0\\
        \pi^{m(m+1)}_\alpha(x_\alpha^0, R_{\alpha,-\alpha}^{m})&=\left(x_\alpha^0, \pi^{(m-1)m}_{-\alpha}(R^m_{\alpha,-\alpha})\right)
    \end{align*}
    Then $\pi^{mn}(x)=(\pi^{mn}_\alpha(x_\alpha))$ is a family of connecting morphisms under which $(X^n)$ forms an inverse family. The inverse limit of this family $X:=\varprojlim X^n$ is called the \emph{consentification} of $X^0$.
\end{dfn}

\begin{dfn}[Indexing consentified choices]
    \label{dfn:index}
    Write the consentified choice space $X$ in the canonical representation of the inverse limit. Then: 
    \begin{enumerate}
        \item $X_\alpha:=\set{(x_\alpha^n):(x^n)\in X}$.
        \item For $x_\alpha\in X_\alpha$, $x_{\alpha,0}:=x_{\alpha,0}^n$, and $x_{\alpha,\beta} := (x^{n}_{\alpha,\beta})$.
        \item For $x_\alpha\in X_\alpha$, $x_\beta\in X_\beta$, we say ``$x_\alpha$ forbids $x_\beta$'' if $\forall n, x^n_\beta\in x^{n+1}_{\alpha,\beta}$.
        \item We denote by $x_{\alpha,\beta}$ the set of $x_\beta$ forbidden by $x_\alpha$.
    \end{enumerate}
\end{dfn}

The purpose of Def~\ref{dfn:index} is to allow us to pretend that $X_\alpha=X_\alpha^0\times \prod_\beta 2^{X_\beta}$ and use notation to such effect, even though this is mathematically impossible (by Cantor's argument). Intuitively, one may factor $2^{X_\beta}$ into the product of Boolean functions on singletons -- these goods are called \emph{consent}, and a quantity of 1 represents non-consent to a particular choice. 

\begin{rmk*}[Def~\ref{dfn:consent} and Cantor]
    It may seem at first glance that despite the formalism, our definition still violates Cantor's argument: if we can ``pretend'' that $X_\alpha = X_\alpha^0\times \prod_\beta 2^{X_\beta}$, that means there is still a natural function from $2^{X_\beta}$ to $X_\alpha$ -- explicitly, for some $R_\beta\subseteq X_\beta$ we may define the ``forbidding choice'' $F_\alpha(R_\beta) = ((x^0_\alpha, x^n_{-\beta}, R_\beta^{n-1}))$ where $R_\beta^n$ is the image of $R_\beta$ under co-ordinate projection and $x_\alpha$ is some arbitrary member of $X_\alpha$ of which we keep the other co-ordinates. However, this function is not an injection, and the reason for this is somewhat delicate: there exist $R_\beta$ such that $F_\alpha(R_\beta)$ forbids choices not in $R_\beta$, i.e. $F_\alpha(R_\beta)_{,\beta}\ne R_\beta$. To see this, let $x_\beta=(x_\beta^n)\in X_\beta$ and construct for each $m$ an $x_{\beta(m)}\in X_\beta$ such that $x_{\beta(m)}^n=x_{\beta}^n$ for $m\le n$ but $x_{\beta(m)}\ne x_\beta $. Then $x_\beta$ is forbidden by $F_\alpha(\set{x_{\beta(m)}})$, and thus $\set{x_{\beta(m)}}$ cannot be a set of the form $x_{\alpha,\beta}$.
\end{rmk*}

\begin{lem}[Topology of $X_\beta$]
    \label{lem:top}
    Let $R\subseteq X_\beta$. Then we define the ``closure'':
    \begin{equation*}
        \cl(R)=\set{y\in X_\beta\mid \forall n, \exists y'\in R, y'_{n}=y_n}
    \end{equation*}
    Then $\cl$ satisfies the Kuratowski closure axioms, hence defining a topology on $X_\beta$. Furthermore, $X_\beta$ is Hausdorff and first-countable. 
\end{lem}

\begin{proof}
    Trivial.
\end{proof}

\begin{cor}
    \label{lem:closed}
    It is precisely the closed sets that may be forbidden (equivalently, open sets are those which may be consented to) -- where $\phi_\beta$ is the topology on $X_\beta$ expressed in terms of closed sets:
    \begin{equation*}
        X_\alpha = X^0_\alpha \times \prod_{\beta} \phi_\beta
    \end{equation*}
    That is: any choice in $X_\alpha$ may be uniquely determined by a choice in $X_\alpha^0$ and closed sets from each $X_\beta$.
\end{cor}

Our definitions so far say nothing of an agent actually possessing a right to non-consent (a non-zero quantity of consent) to any choice; this is defined in Def~\ref{dfn:rights}.

\begin{dfn}[Rights structure]
    \label{dfn:rights}
    Given a consentified choice space $X$, let $R$ be a collection of subsets $R_{\alpha\beta}\subseteq X_\beta$ with utility function satisfying $x_{\alpha,\beta}\nsubseteq R_{\alpha\beta}\implies U_\alpha(x)=\impossible$ (that is: an agent cannot forbid something it has no right to forbid). Then $R$ is called a \emph{rights structure}.
\end{dfn}

\begin{dfn}[Non-Aggression Principle]
    \label{dfn:nap}
    Given the construction in Def~\ref{dfn:consent}, a utility function $U_\alpha:X\to\R$ is called ``non-aggressive'' if it returns $\impossible$ for $x\in X$ such that $\exists\beta,x_\alpha\in x_{\beta,\alpha}$.
\end{dfn}

Our definition of rights is normative, rather than positive: a rights structure cannot be inferred from observations about a real society (unless agents are assumed to be non-aggressive). The notion can be made concrete through a definition of a \emph{government}, that may enforce a certain rights structure.

\begin{dfn}[Government]
    \label{dfn:government}
    For some list of subsets $(S_\alpha\subseteq X_\alpha:\alpha)$, an agent $\gamma$ is called a \emph{government} with penal code $S$ if:
    \begin{itemize}
        \item For every $\alpha$ there exists a Boolean good $p_{\alpha}$ (called ``punishment'') of $X_\gamma$ such that the utility of $\alpha$ is $\impossible$ if $p_\alpha=1$.
        \item The utility of $\gamma$ is $\impossible$ if $\exists\alpha$, $x_\alpha\in S_{\alpha}$ and $p_\alpha\ne 1$.
    \end{itemize}
\end{dfn}

\begin{exl}[Perfect minarchy]
    \label{exl:minarchy}
    A minarchy on some rights structure $R$ is an enforcer $\gamma$ whose penal code is the union of non-consent sets: $S_{\alpha}=\bigcup_\beta x_{\beta, \alpha}$.
\end{exl}

More generally, Def~\ref{dfn:nap}, Ex.~\ref{exl:minarchy} are examples of ways to enforce the assumption that the rights structure is \emph{non-attenuated}. The concept of non-attenuation was first formally discussed by \cite{cheung, demsetz}, and is a requirement for efficiency. 

\section{Prices and exchange}
\label{sec:equilibrium}

We provide a construction of an exchange economy as an abstract economy with a rights structure -- this is an alternative to the standard construction \cite{arrow} of an abstract economy with a market maker. 

\begin{dfn}[Exchange economy as a consentified economy]
    \label{dfn:pure}
    Consider a pure exchange economy with agents $A$, $l$ goods, consumption sets $\bar{X}_\alpha\subseteq\R^l$, utility functions $\set{\bar{U}_\alpha:\bar{X}_\alpha\to\R\mid \alpha\in A}$ and initial endowments $w_\alpha\in \Rp^l$. Define an abstract economy with agents as in $A$, and $\forall\alpha$, a choice set $X^0_\alpha:=\set{x_\alpha:A\to\bar{X}_\alpha}$ and a utility function $U^0_\alpha(x):=\bar{U}_\alpha\left(\sum_{\beta\in A}{x_\alpha(\beta)}\right)$. We then define $X_\alpha$ as in Def~\ref{dfn:consent} and define a utility function over it:
    
    \begin{equation*}
        U_\alpha(x)=
        \begin{cases}
            \impossible & \text{if }\exists y \in x_{\alpha,\beta}, y_0(\alpha)=0 \\
            \impossible & \text{if }\exists\beta,x_\alpha\in x_{\beta,\alpha} \\
            \impossible & \text{if }\exists i\le l, \sum_{\beta\in A}{x_{\beta,0}(\alpha)_i} > w_{\alpha,i}\\
            U^0_\alpha(x_0) & \text{else}
        \end{cases}
    \end{equation*}
\end{dfn}

The three support conditions are of crucial importance: the first defines the rights structure $R_{\alpha\beta}=\set{y\in X_\beta\mid y_\alpha\ne 0}$, the second requires that this rights structure is non-attenuated, and the third imposes a budget constraint.

We do not necessarily have the existence of an equilibrium for an abstract economy with rights, as a consentified choice space is not necessarily compact, etc. In fact, it is not even true that an equilibrium of an exchange economy is an equilibrium of the corresponding consentified economy, as in the absence of perfect competition, an economy may have differential pricing. Instead, we demonstrate some elementary results about pricing in a consentfied economy.

\begin{dfn}(Price)
    \label{dfn:price}
    Consider the construction in Def~\ref{dfn:pure}. For a price vector $p\in\R^l$, we define the corresponding ``price'' in an abstract economy:
    \begin{equation*}
        R_\beta(p)=\set{y\in X_\beta\mid \set{x\in X_\alpha\mid p\cdot x_0(\beta) \ge p\cdot y_0(\alpha)} \subseteq y_{\beta,\alpha}}
    \end{equation*}
    A choice $x_\alpha$ such that each $x_{\alpha,\beta}$ is a price is called a ``pricey choice''.
\end{dfn}

\begin{lem}
    \label{lem:pclosed}
    A price as in Def~\ref{dfn:price} is a closed set under the topology defined in Lemma~\ref{lem:closed}.
\end{lem}
\begin{proof}
    Trivial.
\end{proof}

\begin{lem}[Prices are good]
    \label{lem:pbest}
    If all $x_{-\alpha}$ are pricey, then for any $x_\alpha$, there is a pricey $x'_\alpha$ such that $U_\alpha(x'_\alpha)\ge U_\alpha(x_\alpha)$. 
\end{lem}
\begin{proof}
    Since $x_{-\alpha}$ are fixed, $U_\alpha$ depends on $x_{\alpha,-0}$ only if it falls into one of the first two support conditions. It suffices to show that if there is an $x_{\alpha,-0}$ that doesn't fall into any of these support conditions, then there is a pricey $x_{\alpha,-0}$ that doesn't fall into any of these support conditions.
    
    For a price $x_{\alpha,\beta}$ to contain a $y$ such that $y_0(\alpha)=0$, we must have $X_\alpha\subseteq y_{\beta,\alpha}$, which is not possible if $y_{\beta,\alpha}$ is a price. Thus a pricey choice never violates the first condition, and it suffices to show there is a pricey choice that doesn't violate the second condition. 
    
    If $x_{\beta,\alpha}$ is a price, then it it is of the form 
    
    \begin{equation*}
        \set{x\in X_\alpha\mid\set{y\in X_\beta\mid p\cdot y_0(\alpha)\ge p\cdot x_0(\beta)} \subseteq x_{\alpha,\beta}}
    \end{equation*}
    
    If there is an $x_\alpha$ that doesn't violate the second condition, then that means that for each $\beta$ there is a $y\in X_\beta$ such that $p\cdot y_0(\alpha)\ge p\cdot x_{\alpha,0}(\beta)$ and $y\notin x_{\alpha,\beta}$. We then choose any $p$ such that $R_\beta(p)\subseteq x_{\alpha, \beta}$, which is always possible.
\end{proof}

\section{Future work}
\label{sec:conclusion}

While we have explored the mathematical theory of economies with a rights structure in some detail, much work remains to be done on the dynamical properties of such economies and its price theory. The results we have demonstrated with regards to price equilibria are relatively elementary; it is required to study the existence, physical relevance and relative (in)-efficiency of non-price equilibria in consentified economies arising from exchange economies, as well as in consentified economies in general.

Of particular importance is the production exchange economy as detailed in \cite{arrow}. Production involves the concept of \emph{transferable rights} \cite{coase}, which is among the conditions for a non-attenuated rights structure in property rights theory \cite{cheung}\cite{demsetz}. Under our current formulation, a transfer of rights involves a \emph{transformation} of rights: for example, the possession by $\alpha$ of a ``right to till some land'' entails the forbiddance of $\gamma$ from forbidding $\alpha$ from tilling said land, yet the possession by $\beta$ of a sementically identical right entails the forbiddance of $\gamma$ from forbidding $\beta$ from tilling said land. It is not yet clear if this requires a further reformulation of our definition of rights.

By establishing a mathematical foundation for rights structures, we also pave the way for research on the implications of property rights theory for welfare economics -- the relativity of an allocation to a rights structure is analogous to the aggregation problem of utility in welfare economics. It is worth studying if particular rights structures can be shown to maximize particular measures of aggregate welfare.

\printbibliography

\end{document}